\documentclass{lmcs}

\usepackage{enumerate}
\usepackage{hyperref}
\usepackage{amssymb,latexsym,bm}

\theoremstyle{plain}

\newcommand\twoheaddownarrow{\mathord{\rotatebox{90}{$\twoheadleftarrow$}}}
\def\kurung#1{\left(#1\right)}
\def\himp#1{\{#1\}}

\def\t#1{\text{#1}}
\def\mc#1{\mathcal{#1}}
\def\sub{\subseteq~\!\!\!}

\def\com{\t{c}}

\newcommand{\cl}{\operatorname{cl}}

\newcommand{\Irr}{\operatorname{Irr}}

\newcommand{\bigsup}{\bigvee}
\let\olddownarrow\downarrow 
\renewcommand{\downarrow}{\mathord{\olddownarrow}}

\let\olduparrow\uparrow 
\renewcommand{\uparrow}{\mathord{\olduparrow}}
\newcommand{\lra}{\longrightarrow}

\begin{document}

\title[Strong completion of spaces]{Strong completion of spaces}

\author[H. Andradi]{Hadrian Andradi}	
\address{National Institute of Education, Nanyang Technological University, 1 Nanyang Walk, Singapore 637616}	
\email{hadrian.andradi@gmail.com}  
\thanks{The first author is supported by Nanyang Technological University Research Scholarship (RSS)}	

\author[W. K. Ho]{Weng Kin Ho}	
\address{National Institute of Education, Nanyang Technological University, 1 Nanyang Walk, Singapore 637616}	
\email{wengkin.ho@nie.edu.sg}  





\begin{abstract}
  \noindent A non-empty subset of a topological space is irreducible if whenever it is covered by the union of two closed sets, then already it is covered by one of them.  Irreducible sets occur in proliferation: (1) every singleton set is irreducible, (2) directed subsets (which of fundamental status in domain theory) of a poset are exactly its Alexandroff irreducible sets, (3) directed subsets (with respect to the specialization order) of a $T_0$ space are always irreducible, and (4) the topological closure of every irreducible set is again irreducible.  In recent years, the usefulness of irreducible sets in domain theory and non-Hausdorff topology has expanded.  Notably, Zhao and Ho (2009) developed the core of domain theory directly in the context of $T_0$ spaces by choosing the irreducible sets as the topological substitute for directed sets.  Just as the existence of suprema of directed subsets is featured prominently in domain theory (and hence the notion of a dcpo -- a poset in which all directed suprema exist), so too is that of irreducible subsets in the topological domain theory developed by Zhao and Ho.  The topological counterpart of a dcpo is thus this: A $T_0$ space is said to be \emph{strongly complete} if the suprema of all irreducible subsets exist.  In this paper, we show that the category, $\mathbf{scTop^+}$, of strongly complete $T_0$ spaces forms a reflective subcategory of a certain lluf subcategory, $\mathbf{Top^+}$, of $T_0$ spaces.
\end{abstract}

\maketitle

\section{Introduction}
\label{sec: intro}
The story we are about to tell in this paper involves two kinds of mathematical structures: domains and $T_0$ spaces.
On one hand, domains are partially ordered sets (posets, for short) that have structures rich enough to describe approximation and convergence.  On the other other, $T_0$ spaces are topological spaces which have minimal separation axiom in that every distinct pair of points $x \neq y$ can be distinguished by the existence of an open that contains one point but not the other.

\paragraph{\textbf{Domains.}} In domain theory, the abstract process of producing increasingly accurate approximations to an element is modelled by the directed sets.  A subset $D$ of a poset $P$ is \emph{directed} if for every two elements $d_1$ and $d_2 \in D$ there exists $d_3 \in D$ such that $d_1 \leq d$ and $d_2 \leq d$. Indeed the directed subsets play an indispensable role in domain theory because many fundamental domain-theoretic concepts are formulated using them.  One such notion is that of the \emph{directed complete partial orders} (\emph{dcpo's}, for short) --  these are posets of which every directed subset has supremum;   the justification being that directed sets model the phenomenon of approximation should converge to their suprema.  We shall make this notion of convergence more precisely in a moment, but first we need another important domain-theoretic definition which makes essential use of directed sets, i.e., the \emph{way-below} relation $\ll$.
Here is its definition: $x \ll y$ if for every directed subset $D$ of a dcpo $P$, whenever $\bigsup D \geq y$ there is already $d \in D$ such that $d \geq x$.  Denoting the set $\{y \in P \mid y \ll x\}$ by $\twoheaddownarrow x$, a dcpo $P$ is a \emph{domain} if for every $x \in P$ the set $\twoheaddownarrow x$ is directed and has supremum equals to $x$.

Now we return to elaborate on the convergence we had in mind concerning directed sets, and the precise formulation is given by: A net $(x_i)_{i \in I}$ converges to $y$, written as $(x_i) \lra_d y$, in the dcpo $P$ if there exists a directed set of eventual lower bounds of $(x_i)_{i \in I}$ whose supremum is above $y$.  A convergence structure $(P,\lra)$ on a dcpo is said to be \emph{topological} if there exists a topology on it such that $(x_i) \lra y$ if and only if whenever an open set $U$ contains $y$ then $U$ contains $(x_i)_{i \in I}$ eventually.
Domains are very special dcpo's whose characterized is succinctly given by: The convergence structure $(P,\lra_d)$ on a dcpo $P$ is topological if and only if $P$ is a domain.  The topology which makes this convergence topological is the \emph{Scott topology} -- the most prominently featured topology in domain theory.  Again, the Scott opens which constitute the topology are defined via directed sets: $U \subseteq P$ is \emph{Scott open} if $U$ is Alexandroff open (i.e., upper) and inaccessible by (existing) directed suprema, i.e., whenever $\bigsup D \in U$ one has $D \cap U \neq \emptyset$.  Crucially, the Scott topology on a poset turns it into a $T_0$ space.

\paragraph{\textbf{$T_0$ spaces.}}
$T_0$ spaces can in fact be partially ordered: given any $T_0$ space $X$, one defines the specialization order $\leq_X$ on it as follows: $x \leq_X y$ if every open $U$ that contains $x$ must contain $y$.  Whenever an order-theoretic concept is mentioned in the context of a $T_0$ space, its meaning is interpreted with the specialization order in mind.  For example, a directed subset $D$ of a $T_0$ space $X$ would mean that the set $D$ is directed with respect to the specialization order $\leq_X$.  Note that the specialization orders induced by the Alexandroff and Scott topologies on a poset $P$ both coincide with the underlying order of $P$.

Through the topological lens, directed subsets have another guise.  The salient topological attribute is that of irreducibility.  A non-empty subset $F$ of a $T_0$ space $X$ is said to be \emph{irreducible} if whenever $F \subseteq A \cup B$ for two closed sets $A$ and $B$ then either $F \subseteq A$ or $F \subseteq B$.  Irreducible subsets occur ubiquitously, in that every singleton subset is irreducible with respect to any topology and the closure of irreducible subsets is again irreducible.  Directed sets and irreducible sets are closely linked in nifty ways: (1) Directed subsets of a poset $P$ are exactly its Alexandroff irreducible subsets, and (2) directed subsets of a $T_0$ space are always irreducible.  In general, an irreducible set of a $T_0$ space need not be directed (see~\cite{johnstone82}).  Thus, the notion of an irreducible set can be seen as a topological generalization of directed set.

Just as directed sets are to domains, so are irreducible sets to $T_0$ spaces.  Irreducible sets play an important role in the topology. It is well-known that every topological space $X$ gives rise to its lattice of opens, $\mathcal{O}(X)$, whose elements are open sets of $X$ ordered by inclusion.  If $X$ and $Y$ are homeomorphic spaces, then their lattice of opens $\mathcal{O}(X)$ and $\mathcal{O}(Y)$ are isomorphic.  Two topological spaces $X$ and $Y$ are said to be \emph{lattice equivalent} if their lattice of opens are isomorphic.  It turns out that lattice equivalent spaces are not necessarily homeomorphic unless both the spaces are sober.  A space is \emph{sober} if every closed irreducible set is the closure of a unique singleton -- which is where irreducible sets get involved.  Because of the many pleasing properties that they enjoy, sober spaces have a special place in the study of domain theory and non-Hausdorff topology. For instance, every sober space is a dcpo with respect to its specialization order.

\paragraph{\textbf{Domains in $T_0$ space.}} In recent years, many theorems in domain theory have their analogues in the wider context of $T_0$ spaces.  Here is an example that is particularly important to the development of this paper.  As is known, not every poset is a dcpo.  Given a poset $P$, Zhao and Fan~\cite{zhaofan10} showed that one can complete $P$ to a dcpo $E(P)$ in the sense that there exists a Scott continuous embedding $P \hookrightarrow E(P)$ such that every Scott continuous mapping $f:P \lra Y$, where $Y$ is a dcpo, can be extended to a unique Scott continuous mapping $\hat{f} : E(P) \lra Y$.  The dcpo $E(P)$ is called the canonical dcpo-completion of the poset $P$.
This domain-theoretic result has recently been generalized to the realm of $T_0$ space by Keimel and Lawson (\cite{keimellawson09}) as we shall explain.  A $T_0$ space is a \emph{d-space} (also known as a \emph{monotone convergence space}) if it is a dcpo and every open open is Scott open.  Every domain endowed with its Scott topology is a d-space. Thus, a d-space can be thought of as a $T_0$-space that behaves like a domain since every directed subset has a supremum and the directed set as a net converges to its supremum.  Essentially, Keimel and Lawson showed in~\cite{keimellawson09} that there is also a canonical D-completion for a $T_0$ space to a d-space in much the same spirit as a dcpo-completion.

Another example of this kind of generalization is that of the order-theoretic Rudin's lemma, which is central to the theory of quasicontinuous domains, to its topological version (\cite{heckmannkeimel13}).  Instances of such upgrading of domain-theoretic results to the context of $T_0$ spaces occur in such proliferation that J. D. Lawson in his plenary lecture at the Sixth International Symposium in Domain Theory at Changsha, China drew up a systematic scientific program that investigates how domain theory might manifest directly in $T_0$ spaces.

This programmatic call was quickly responded by several domain-theorists.  Notably, D. Zhao and the second author saw that irreducible sets are a topological generalization of directed sets, and systematically substituted irreducible sets for directed sets in many of their key notions~(\cite{zhaoho15}).  In that work, they developed the core of domain theory \emph{directly} in the context of $T_0$ spaces.  Consequently, several landmark results in domain theory find their analogues in the $T_0$ space setting.  We shall be revisit some of the key concepts and results derived in~\cite{zhaoho15}.

This paper is a continuation of the work started in \cite{zhaoho15} that in particular extends the completion result obtained in \cite{zhaofan10}.  In domain theory, not directed subset of a poset has supremum; likewise for $T_0$ spaces, not irreducible subset of a $T_0$ space has supremum.  A $T_0$ space of which every irreducible subset has a supremum is called a \emph{strongly complete space}\footnote{This term was first coined in~\cite{hojungxi16} but in a more specific context of Scott topology}.  In this paper, we give a canonical \emph{strong completion} of a $T_0$ space a l\'{a} Zhao and Fan~(\cite{zhaofan10}), and thereby establish that the category $\mathbf{scTop^+}$ of strongly complete $T_0$ spaces is a reflective full subcategory of the category $\mathbf{Top^+}$ of $T_0$ spaces whose morphisms are those continuous maps that preserve existing suprema of irreducible sets.  Thus, our results give yet another generalization of Zhao and Fan~\cite{zhaofan10},  following the approach suggested by Zhao and Ho~\cite{zhaoho15}.

\paragraph{\textbf{Organization.}}  We organize this paper as follows.  Section~\ref{sec: prelim} gathers all necessary preliminaries for the development of our results.  Because many of these originate from \cite{zhaoho15}, no proofs will be given in that section.  We introduce the notions of strongly complete spaces and mappings that preserve the exisiting suprema of irreducible sets in Section~\ref{sec: strongly complete space}.  The main body of work in this paper is performed in Section~\ref{sec: strong completion}, where the canonical strong completion of a space is given.  Some final remarks are made in Section~\ref{sec: conclusion}.

\section{Preliminaries}
\label{sec: prelim}
In this section, we present some basic definitions, notations, and results which are important in our ensuing development.

Given a topological space $X$, we denote the collection of all open (respectively, closed) sets of $X$ by $\mc{O}(X)$ (respectively, $\Gamma(X)$). When a $T_0$ space $X$ is considered as a poset, the default order is its specialization order $\leq_X$ (or $\leq$ whenever there is no confusion). Throughout our discourse, the term \emph{space} always refers to a $T_0$ topological space.

This paper deals with irreducible sets and it is always handy to have an alternative but equivalent description of these sets.
A subset $F$ of a space $X$ is irreducible if and only if whenever it meets two open sets of $X$ non-emptily it also meets their intersection non-emptily.  The set of all irreducible subsets of a space $X$ is denoted by $\Irr(X)$.  An element $F \in \Irr(X)$ is in $\Irr^{+}(X)$ if and only if $\bigsup F \in X$. The following proposition gathers some elementary properties regarding irreducible sets at one place:
\begin{prop}
\label{prop: elementary of irred}
For any space $X$, it holds that:
\begin{enumerate}
\item every singleton is irreducible;
\item a set is in $\Irr(X)$ if and only if its closure is in $\Irr(X)$;
\item image of an irreducible space under (topological) continuous function is irreducible in the target space;
\item every directed subset of $X$ is irreducible; and
\item for every nonempty subspace $Y$ of $X$, it holds that $\Irr(Y) \sub \Irr(X)$.
\end{enumerate}
\end{prop}

Let $X$ be a space. A subset $U$ of $X$ is said to be \emph{irreducibly open} if (i) $U \in \mc{O}(X)$ and (ii) whenever $\bigsup F \in U$ for some $F \in \Irr^{+}(X)$ it holds that $F \cap U \neq \emptyset$.  The collection of all irreducibly open subsets of $U$ form a topology on $X$ called the irreducibly-derived topology.  The set $X$ endowed with this topology is denoted by $SI(X)$.  Open (respectively, closed) sets in $SI(X)$ are called \emph{$SI$-open sets} (respectively, \emph{$SI$-closed sets}).  For simplicity, we shall write $\mc{O}_{SI}(X)$ instead of $\mc{O}(SI(X))$ to denote the irreducibly-derived topology; likewise $\Gamma_{SI}(X)$ instead of $\Gamma(SI(X))$ for the collection of closed sets in $SI(X)$.

\begin{prop} (\cite{zhaoho15})
Let $X$ be a space. Then the following hold:
\begin{enumerate}
\item For any $x \in X$, $\cl_X(\himp{x}) = \cl_{SI}(\himp{x})$.
\item A subset $C$ of $X$ is closed in $SI(X)$ if and only if $C$ is closed in $X$ and for every $F \in\Irr^{+}(X)$, $F \subseteq C$ implies $\bigsup F \in C$.
\item A subset $U$ of $X$ is clopen in $X$ if and only if it is clopen in $SI(X)$.
\item $X$ is connected if and only if $SI(X)$ is connected.
\end{enumerate}
\end{prop}

On a poset $(P,\leq)$, for any $A \subseteq P$, the set
$\uparrow A$ refers to the set of all elements in $P$ which is below or equal to some element of $A$. The set $\downarrow A$ is defined dually. It can be easily seen that the operator ``$\uparrow$' and ``$\downarrow$'' on the powerset of $P$ are idempotent.

\section{Strongly complete space and $SI^{+}$-continuous function}
\label{sec: strongly complete space}

In \cite{zhaofan10} and \cite{keimellawson09}, a new topology on posets called the \emph{$d$-topology} was introduced, whose definition makes essential use of directed sets.   The $d$-topology is distinct from the previously well-established topologies, i.e., the Scott topology and the Alexandroff topology.  A subset $A$ of a poset $P$ is called \emph{$d$-closed} if for all directed subset $D$ of $P$ such that $D\sub A$, it holds that $\bigsup D\in A$.  In the same spirit, we create the concept of $I$-closed sets and $I$-open sets on a given space $X$ -- which are different from the closed sets in original topology and the irreducibly closed sets given in \cite{zhaoho15}. Unlike the $d$-closed sets, the collection of $I$-closed sets do not define a topology on $X$.

\begin{defi}
Let $X$ be a space. A subset $A$ is \emph{$I$-closed} if for any $F \in \Irr^{+}(X)$ such that $F \sub A$, it holds that $\bigsup F \in A$. A subset of $X$ is \emph{$I$-open} provided its complement is $I$-closed.
\end{defi}

We want to emphasize that the collection of all $I$-open sets of $X$ does not in general form a topology on $X$ as it need not be closed under finite intersection.  In this sense, one immediately recognizes the difference between irreducible sets and directed sets.  However, we still have that the intersection of $I$-closed sets is again $I$-closed, i.e., the collection of $I$-closed sets form a closure system on $X$.

\begin{prop}
Let $\himp{A_j \mid j \in J}$ be a collection of $I$-closed subsets of a space $X$. Then $\bigcap_{j \in J} A_j$ is $I$-closed.
\end{prop}

\proof
Let $F\in\Irr^{+}(X)$ such that $F\sub \bigcap_{j\in J} A_j$. Then $F\sub A_j$ for all $j\in J$. Since $A_j$ is $I$-closed, then $\bigsup F\in A_j$. This implies $\bigsup F\in \bigcap_{j\in J} A_j$. Hence $\bigcap_{j\in J} A_j$ is $I$-closed.
\qed

Let $X$ be a space and $A \sub X$.
We invent some new notations:
\begin{enumerate}
\item $\Delta(X) := \{U \sub X \mid U \t{ is }I\t{-open}\}$,
\item $\Theta(X) := \{C \sub X \mid C \t{ is }I\t{-closed}\}$, and
\item $\cl_I(A) :=
      \bigcap \{C \sub X\ C\in \Theta(X)\t{ and }A\sub C\}$. The set $\cl_I(A)$ is called the \emph{$I$-closure} of $A$.
\end{enumerate}
Here, notice that the notion of $I$-closed (respectively, $I$-open) sets are defined using the underlying topology.
This differs from the concepts of $d$-closed and $d$-open sets which are intrinsic order-theoretic notions.

The following lemma is immediate.
\begin{lem}
The following hold for any space $X$:
\begin{enumerate}[\normalfont (1)]
\item Every upper set is $I$-closed.
\item Every principal ideal $\downarrow x$ is closed, $I$-closed, and $SI$-closed. Further, it holds that $\cl\kurung{\himp{x}}=\cl_{SI}\kurung{\himp{x}}=\downarrow x$.
\item A subset $U$ of $X$ is in $\Delta X$ if and only if for all $F\in\Irr^{+}(X)$ such that $\bigsup F\in U$ it holds that $F\cap U\neq\emptyset$.
\item $\mc{O}(X)\cap \Delta(X)= \mc{O}_{SI}(X)$.
\item $\Irr(X)\sub \Irr\kurung{SI(X)}$.
\end{enumerate}
\end{lem}

Analogous to directed completeness defined on poset, we define another completeness on a space; this is done by replacing directed sets with irreducible ones.

\begin{defi}
A space $X$ is \emph{strongly complete} (or \emph{sc}, for short) if $\Irr(X)=\Irr^{+}(X)$.  In other words, a space is strongly complete if every irreducible subset has a supremum.  A subspace space $Y$ of $X$ is called an \emph{sc-subspace} of $X$ if $Y$ is an sc-space with respect to the relative topology.
\end{defi}

\begin{prop}
Every sc-space is a dcpo with respect to its specialization order.
\end{prop}
\begin{proof}
Let $D$ be a directed subset of and sc-space $X$.  Then $D$ is an irreducible subset
of $X$ and by the strong completeness of $X$, $\bigsup D$ exists.
\end{proof}

\begin{cor}
Every sober space is a dcpo with respect to its specialization order.
\end{cor}
\begin{proof}
Clear since every sober space is an sc-space.
\end{proof}

It is easy to establish the following:
\begin{lem}\label{lem:irreducible_in_subspace}
Let $X$ be a space and $Y$ be a subspace of $X$. Then
\begin{enumerate}[\em(i)]
\item $\Irr(Y) = \himp{F\in \Irr(X)\mid F\sub Y}$.
\item For all $a,~b\in Y$, $a \leq_{Y} b$ if and only if $a \leq_{X} b$.
\end{enumerate}
\end{lem}
%
%

The following characterization of an sc-subspace of a given sc-space is a direct corollary of the previous lemma.
\begin{cor}\label{cor:ic_subspace_of_ic_space}
Let $X$ be an sc-space and $Y$ be its subspace. Then $Y$ is $I$-closed implies $Y$ is an sc-subspace.
\end{cor}
\proof
Let $F \in\Irr(Y)$. By Lemma \ref{lem:irreducible_in_subspace}, $F \in \Irr^{+}(X)$. Since $Y$ is $I$-closed, $\bigsup_X F \in Y$. If $u$ is an upper bound of $F$ in $Y$ then it is an upper bound of $F$ in $X$. Hence $\bigsup_X F \leq u$. This implies $\bigsup_X F = \bigsup_Y F$, from which we conclude that $Y$ is an sc-subspace.
\qed

Though several properties of the irreducibly derived topology, $SI(X)$, of $X$ were investigated in \cite{zhaoho15}, nothing was mentioned about the continuity of a function with respect to the irreducibly-derived topologies.   We fill in this gap by giving the following new definitions:
\begin{defi}
Let $X$ and $Y$ be spaces. The function $f:X\lra Y$ is said to be
\begin{enumerate}[(1)]
\item \emph{$I$-continuous} if for all $B\in \Theta(Y)$ it holds that $f^{-1}(B)\in \Theta(X)$;
\item \emph{$SI$-continuous} if the function $f:SI(X)\lra SI(Y)$ is continuous; and
\item \emph{$SI^{+}$-continuous} if $f$ is both continuous and $SI$-continuous.
\end{enumerate}
\end{defi}

Because every identity function is $SI$-continuous and composition of two $SI$-continuous functions is again $SI$-continuous, it is legitimate to form the category \textbf{Top}$\bm{_0^+}$ whose objects are the $T_0$ spaces and morphisms the $SI^+$-continuous functions.

It follows immediately from the definition that $f$ is $I$-continuous if and only if the pre-image of any $I$-open set along $f$ is $I$-open. Indeed, the name ``$I$-continuous'' is so given because $I$-continuous maps behave like continuous ones with respect to the $I$-open sets. The following basic remarks regarding continuous function are necessary for our upcoming development.

\begin{rem}\label{I-cts}
Let $f:X\lra Y$ be a continuous function.
\begin{enumerate}[\normalfont (1)]
\item For any subset $A$ of $X$, it holds that $f\kurung{\cl(A)}\sub\cl\kurung{f(A)}$.
\item If $F\in\Irr(X)$ then $f(F)\in \Irr(Y)$.
\item If $f:X\lra Y$ is continuous and $I$-continuous then it is $SI$-continuous.
\end{enumerate}
\end{rem}

\begin{lem}\label{lem:continuous_w.r.t._irreducibly-derived_topology_implies_monotonic}\hfill
\begin{enumerate}[\em(i)]
\item If $f:X\lra Y$ is continuous then $f$ is monotonic.

\item If $f:X\lra Y$ is $SI$-continuous then $f$ is monotonic.

\item If $f:X\lra Y$ is $SI$-continuous  then for any $F\in\Irr^{+}(X)$, $\bigsup f(F)$ exists in $Y$ and is equal to $f\kurung{\bigsup F}$.

\item If $f:X\lra Y$ is monotonic and $I$-continuous then for any $F\in\Irr^{+}(X)$ it holds that $f\kurung{\bigsup F}=\bigsup f(F)$.
\end{enumerate}
\end{lem}

\proof\hfill
\begin{enumerate}[\normalfont (i)]
\item Let $a\leq_X b$. Suppose in contrary that $f(a)\nleq_{Y} f(b)$. Then $f(a)\in V:=Y\setminus \downarrow f(b)$. Since $\downarrow f(b)$ is closed in $Y$,  $f^{-1}(V)$ is open in $X$, hence upper. As $a\in f^{-1}(V)$, we have $b\in f^{-1}(V)$. This leads to a contradiction.

\item Considering the fact that every principal ideal is $SI$-closed and $SI$-closed set is a closed set, the proof is analogue to that in part (i).

\item By monotonicity of $f$ we have $f\kurung{\bigsup F}\geq f(a)$ for all $a\in F$, hence $f\kurung{\bigsup F}$ is an upper bound of $f(F)$. Let $u$ be any upper bound of $f(F)$. Since $\downarrow u$ is $SI$-closed, by assumption $f^{-1}\kurung{\downarrow u}$ is $SI$-closed. As $f(F)\sub \downarrow u$, $F\sub f^{-1}\kurung{f(F)}\sub f^{-1}\kurung{\downarrow u}$. By $SI$-closedness of $f^{-1}\kurung{\downarrow u}$, we have $\bigsup F\in f^{-1}\kurung{\downarrow u}$. Therefore $f\kurung{\bigsup F}\leq u$ which is the desired result.

\item Using the monotonicity of $f$ and considering the fact that every principal ideal is $I$-closed, the proof is analogue to that in part (iii).\qed
\end{enumerate}

Every Scott continuous function preserves existing directed suprema.  In the presence of the continuity assumption,
an $SI$-continuous function can also be characterized by the fact that it preserves the existing irreducible suprema.
\begin{lem}\label{lem:equiv_condition_of_SI-continuity}
If $f: X \lra Y$ is a continuous function, then the following are equivalent:
\begin{enumerate}[\em(1)]
\item The function $f$ is $I$-continuous.
\item The function $f$ is $SI$-continuous;
\item For any $F\in\Irr^{+}(X)$ it holds that $f\kurung{\bigsup F}=\bigsup f(F)$.
\end{enumerate}
\end{lem}
\proof
(1) implies (3) and (2) implies (3) are true by Lemma \ref{lem:continuous_w.r.t._irreducibly-derived_topology_implies_monotonic}. (1) implies (2) is immediate by Remarks \ref{I-cts}. We now show that (3) implies (1) is true. Let $A$ be $I$-closed in $Y$ and $F\in\Irr^{+}(X)$ with $F\sub f^{-1}(A)$. By continuity of $f$, $f(F)$ is an irreducible set in $Y$ contained in $A$. Hence, by $I$-closedness of $A$, $\bigsup f(F)=f\kurung{\bigsup F}\in A$. It yields $\bigsup F\in f^{-1}(A)$, and the proof is complete.
\qed

\begin{cor}\label{cor:I_closure_of_SI+_image}
If $f:X\lra Y$ is $SI^{+}$-continuous then for all $A\sub X$ it holds that
$$f\kurung{\cl_{I}(A)}\sub \cl_{I}\kurung{f(A)}.$$
\end{cor}

\proof
By Lemma \ref{lem:equiv_condition_of_SI-continuity} we have that $f$ is $I$-continuous. Suppose $f\kurung{\cl_{I}(A)}\nsubseteq \cl_{I}\kurung{f(A)}$. Then there exists $x\in \cl_I(A)$ such that $f(x)\notin \cl_I(f(A))$. Hence there exists $M\in \Theta(Y)$ such that $f(A)\sub M$ and $f(x)\notin M$. Since $f$ is $I$-continuous, we have that $f^{-1}(M)\in\Theta(X)$ such that $A\sub f^{-1}\kurung{f(A)}\sub M$ and $x\notin f^{-1}(M)$. It contradicts $x\in \cl_I(A)$. Therefore $f\kurung{\cl_{I}(A)}\sub \cl_{I}\kurung{f(A)}$.
\qed

Corollary \ref{cor:I_closure_of_SI+_image} above is an immediate consequence of Lemma \ref{lem:equiv_condition_of_SI-continuity}.
To end this section, we note the following:
\begin{prop}\label{prop:equality_of_function_w.r.t._cl_I}
If $Y$ is a subspace of a space $X$ and $f,~g:\cl_{I}(Y)\lra Z$ are $SI$-continuous functions satisfying $f\mid_{Y}=g\mid_{Y}$ then $f=g$.
\end{prop}

\proof
If $F\in\Irr^{+}(X)$ such that $F\sub \himp{x\in \cl_{I}(Y)\mid f(x)=g(x)}$ then $f(F)=g(F)$ and by Lemma \ref{lem:continuous_w.r.t._irreducibly-derived_topology_implies_monotonic} we have $f\kurung{\bigsup F}=\bigsup f(F)=\bigsup g(F)=g\kurung{\bigsup F}$.
Thus $\bigsup F\in \himp{x\in \cl_{I}(Y)\mid f(x)=g(x)}$. We have that $\himp{x\in \cl_{I}(Y)\mid f(x)=g(x)}$ is an $I$-closed set containing $Y$. Since $\cl_{I}(Y)$ is the smallest $I$-closed set containing $Y$, it follows that
$$\cl_I(Y)\sub\himp{x\in \cl_{I}(Y)\mid f(x)=g(x)}\sub\cl_I(Y).$$
\qed

\section{Strong completion}
\label{sec: strong completion}
We have prepared ourselves thus far to now perform the promised construction of the strong completion for a $T_0$ space.
Our strong completion is inspired by the procedure of dcpo-completion given in \cite{zhaofan10}.  Let's familiarize the reader with this concept first.  A \emph{dcpo completion} of a poset $P$ is a dcpo $A$ together with a Scott continuous mapping $\eta_P:P \lra A$, such that for any Scott continuous mapping
$f:P \lra B$ to a dcpo $B$ there exists a unique Scott continuous mapping $\hat{f}:A \lra B$ satisfying the equation $f =\hat{f}\circ \eta_P$.

Mimicking the above formulation, we perform the usual replacement exercise to obtain the following definition:
\begin{defi}
A \emph{strong completion} of a space $X$ is a pair $\kurung{Y,\eta_X}$ where $Y$ is an strongly complete space and $\eta:X \lra Y$ is an $SI^{+}$-continuous function such that for any $SI^{+}$-continuous function $f:X \lra Z$ to a strongly complete space $Z$ there exists a unique $SI^{+}$-continuous function $\overline{f}:Y \lra Z$ satisfying the equation $f = \overline{f}\circ \eta_X$.
\end{defi}

Guided by our categorical instincts, it is not surprising that  the following proposition holds:
\begin{prop}\label{prop:uniqueness_of_i-completion}
The strong completion of a space $X$, if exists, is unique up to homeomorphism.
\end{prop}
%

Categorically, the existence of a strong completion of a $T_0$ space, once established, is equivalent to the fact that the category \textbf{ic-sp} of  sc-spaces and $SI^+$-continuous functions forms a reflective full subcategory of \textbf{Top}$\bm{_0^+}$ since it is clear that the inclusion functor is right adjoint to the functor $\hat{(-)}$ that assigns $X$ to its strong completion $\hat{X}$.

Just as the dcpo completion of a poset produces a new partial order which must be directed complete, our task of producing a strong completion of a space should be a new space which is strongly complete.  To this end, we appeal to the lower Vietoris topology:
\begin{lem}\label{lem:col_of_SI-closed_is_ic}
We endow the set $\Gamma_{SI}\kurung{X}$, i.e., collection of all $SI$-closed subsets of $X$, with topology generated by subbasic open sets of the form
$$\Diamond U :=
\himp{C\in \Gamma_{SI}\kurung{X}\mid C\cap U\neq \emptyset},$$ where $U\in \mc{O}_{SI}\kurung{X}$.
Then
\begin{enumerate}[\em(1)]
\item the specialization order on $\Gamma_{SI}\kurung{X}$ is the inclusion relation, and
\item the space $\Gamma_{SI}\kurung{X}$ is an sc-space.
\end{enumerate}
\end{lem}

\proof\hfill
\begin{enumerate}[(1)]

\item Let $C_1,C_2\in \Gamma_{SI}\kurung{X}$.  If $C_1\sub C_2$, then for any $U\in\mc{O}_{SI}\kurung{X}$ we have  $C_1\cap U\neq\emptyset$ implies $C_2\cap U\neq\emptyset$. Thus $C_1\leq C_2$.

Conversely, suppose $C_1\nsubseteq C_2$. Then $V:=X\setminus C_2\in \mc{O}_{SI}\kurung{X}$ satisfies $C_1\cap V\neq \emptyset$ but $C_2\cap V=\emptyset$. This leads to $C_1\nleq C_2$.

\item If $\mc{F}=\himp{A_i\in \Gamma_{SI}\kurung{X}\mid i\in J}$ we have that
$$\bigcap\himp{C\in \Gamma_{SI}\kurung{X}\mid A_i\sub C \t{ for every }i\in J}$$
is the smallest $SI$-closed set containing any member of $\mc{F}$. Therefore it is the supremum of $\mc{F}$.\qed
\end{enumerate}

Let us denote the strong completion of $X$ by $SC(X)$.
The following theorem generalizes Zhao and Fan's procedure of dcpo-completion (\cite{zhaofan10}):

\begin{thm}\label{th:irr_completion}
Consider $\Psi(X) = \{\cl\kurung{\himp{x}} \mid x \in X\}$ as a subset of $\Gamma_{SI}\kurung{\Psi(X)}$.
Then, the subspace $\cl_{I}\kurung{\Psi(X)}$ of $\Gamma_{SI}\kurung{X}$ is a strong completion of $X$.
\end{thm}

\proof
By Lemma \ref{lem:col_of_SI-closed_is_ic} and Corollary \ref{cor:ic_subspace_of_ic_space} we already know that $\cl_{I}\kurung{\Psi(X)}$ is an sc-space. Let $\eta_X:X\lra \Gamma_{SI}\kurung{X}$ be a function defined by $\eta_X(x)=\cl\kurung{\himp{x}}$ and $U\in \mc{O}_{SI}{(X)}$. We have that
$$\eta_X^{-1}\kurung{\Diamond U\cap \cl_I\kurung{\Psi(X)}}=\himp{x\in X\mid \cl\kurung{\himp{x}}\cap U}=U$$
since $\cl\kurung{\himp{x}}=\downarrow x$ and $U=\uparrow U$. Thus $\eta_X$ is continuous.

Let $F=\himp{x_i\mid i\in J}\in \Irr^{+}(X)$ with $\bigsup F=x$. Since $\eta_X$ is continuous, $\eta_X(F)$ is irreducible in $\cl_{I}\kurung{\Psi(X)}$. Hence $\bigsup\eta_X (F) \in \cl_{I}\kurung{\Psi(X)}$. Let $A$ be any upper bound of $\eta_X(F)$. Then $F\sub A$. By $SI$-closedness of $A$, we have $x\in A$. Therefore $\cl\kurung{\himp{x}}=\downarrow x\sub \downarrow A= A$, implying that $\bigsup \eta_X(F)=\eta_X\kurung{\bigsup F}$. It yields $\eta_X$ is $SI^{+}$-continuous.

Let $Z$ be any sc-space and $f:X\lra Z$ be $SI^{+}$-continuous. We then have $f^{-1}:\Gamma_{SI}\kurung{Z}\lra \Gamma_{SI}\kurung{X}$ is a well-defined function. We define a function $f^{*}:\Gamma_{SI}\kurung{X}\lra \Gamma_{SI}\kurung{Z}$ by $f^{*}(C)=\cl_{SI}\kurung{f(C)}$. We have that for any $\kurung{C,A}\in \Gamma_{SI}\kurung{Z}\times \Gamma_{SI}\kurung{X}$ it holds that
$$f^{*}(C)\sub A\t{ if and only if }C\sub f^{-1}(A)$$
Noting that the specialization order on $\Gamma_{SI}\kurung{Z}$ and $\Gamma_{SI}\kurung{X}$ is inclusion, $f^{*}$ is a left adjoint of $f^{-1}$. Thus $f^{*}$ preserves any supremum, including the supremum of irreducible subsets.
For any $V\in \mc{O}_{SI}(Z)$ we claim that
$$\kurung{f^{*}}^{-1}\kurung{\Diamond V}=\Diamond f^{-1}(V)$$
If $C$ does not belongs to RHS then $f(C)\sub V^{\com}$ which, together with $SI$-closedness of $V^{\com}$, implies $\cl_{SI}\kurung{f(C)}\sub V^{\com}$. Thus $C$ is not a member of LHS. Conversely, if $C$ belongs to RHS then $\cl_{SI}\kurung{f(C)}\cap V\neq \emptyset$ which finishes the proof of our claim. Therefore we have that $f^{*}$ is an $SI^{+}$-continuous function.

For any $x\in X$, we have
$$\himp{f(x)}\sub  f\kurung{\cl_{SI}\kurung{\himp{x}}}\sub \cl_{SI}\kurung{f\kurung{\cl_{SI}\kurung{\himp{x}}}},$$
hence $$\cl_{SI}\kurung{\himp{f(x)}}\sub \cl_{SI}\kurung{f\kurung{\cl_{SI}\kurung{\himp{x}}}}=f^{*}\kurung{\cl_{SI}\kurung{\himp{x}}}.$$
By $SI^{+}$continuity of $f$ we have
$$f\kurung{\cl_{SI}\kurung{\himp{x}}}\sub \cl_{SI}\kurung{\himp{f(x)}},$$
hence
$$f^{*}\kurung{\cl_{SI}\kurung{\himp{x}}}=\cl_{SI}\kurung{f\kurung{\cl_{SI}\kurung{\himp{x}}}}\sub\cl_{SI}\kurung{\himp{f(x)}}$$
Therefore we have that $f^{*}\kurung{\cl\kurung{\himp{x}}}=\cl\kurung{\himp{f(x)}}$.

We define a function $k:\Psi(Z)\lra Z$ by $k\kurung{\cl\kurung{\himp{z}}}=z$. It is continuous since
$$k^{-1}\kurung{U}=\himp{\cl\kurung{\himp{z}}\in\Psi(Z)\mid z\in U}=\Diamond U\cap\Psi(Z)$$
for all $U\in \mc{O}(Z)$.  Let $\mc{F}\in\Irr\kurung{\Gamma_{SI}(Z)}$. Then $\mc{F}=\himp{\cl\kurung{\himp{z_i}}\mid i\in J}$ for some non-empty index set $J$. We have $k(\mc{F})=\himp{z_i\mid i\in J}$. By continuity of $k$ and the fact that $Z$ is an sc-space we have $\bigsup k(\mc{F}):=z\in Z$. For all $i\in J$, $z_i\leq z$ which implies $\cl\kurung{\himp{z_i}}\sub \cl\kurung{\himp{z}}$. Hence $\cl\kurung{\himp{z}}$ is an upper bound of $\mc{F}$. If $A\in \Gamma_{SI}\kurung{Z}$ is an upper bound of $\mc{F}$ then $\cl\kurung{\himp{z_i}}\sub A$ for all $i\in J$, hence the irreducible set $k(\mc{F})$ is a subset of $A$. Thus $SI$-closedness of $A$ gives that $z\in A$. Therefore $\cl\kurung{\himp{z}}=\downarrow z\sub\downarrow A=A$. This yields $\bigsup \mc{F}=\cl\kurung{\himp{z}}\in\Psi(Z)$ which leads to the fact $\Psi(Z)$ is an $I$-closed set in $\Gamma_{SI}(Z)$. This implies $\cl_I\kurung{\Psi(Z)}=\Psi(Z)$. In addition, we have that $k$ is $SI^{+}$-continuous.

By Corollary \ref{cor:I_closure_of_SI+_image}, $SI^{+}$-continuity of $f^{*}$ gives
$$f^{*}\kurung{\cl_I\kurung{\Psi(X)}}\sub \cl_I\kurung{f^{*}\kurung{\Psi(X)}}\sub \cl_I\kurung{\Psi(Z)}=\Psi(Z)$$
We have that the function $\hat{f}:\cl_I\kurung{\Psi(X)}\lra Z$ defined by $\hat{f}=k\circ f^{*}$ is a well-defined $SI^{+}$ continuous function.

For all $x\in X$ we have
$$\kurung{\hat{f}\circ\eta_X}(x)=\hat{f}\kurung{\eta_X(x)}=\hat{f}\kurung{\cl\kurung{\himp{x}}}=k\kurung{f^{*}\kurung{\cl\kurung{\himp{x}}}}=k\kurung{\cl\kurung{\himp{f(x)}}}=f(x)$$
Let $g:\cl_{I}\kurung{\Psi(X)}\lra Z$ be any $SI^{+}$-continuous function satisfying $f=g\circ\eta_X$. For any $\cl\kurung{\himp{x}}\in\Psi(X)$ it holds that $$g\kurung{\cl\kurung{\himp{x}}}=g\kurung{\eta_X(x)}=f(x)=\hat{f}\kurung{\eta_X(x)}=\hat{f}\kurung{\cl\kurung{\himp{x}}}$$
Thus, we have that $\hat{f}$ and $g$ are the same function when we restrict their domain of definition to $\Psi(X)$. By Proposition~\ref{prop:equality_of_function_w.r.t._cl_I}, we conclude that $\hat{f}$ and $g$ are equal, and this gives the uniqueness of $\hat{f}$.
\qed
The following is a direct consequence of Proposition \ref{prop:uniqueness_of_i-completion} and Theorem \ref{th:irr_completion}.

\begin{cor}
Every strong completion of a space $X$ is homeomorphic to $cl_{I}\kurung{\Psi(X)}$.
\end{cor}

\section{Conclusion}
\label{sec: conclusion}
Recent years saw an emphasis on the use of irreducible subsets of a $T_0$ space to investigate the topological and domain-theoretic properties of the space.  A systematic replacement of directed sets in domain theory by irreducible sets allows many of the results in domain theory to be generalized in the context of $T_0$ spaces.  In this paper, we continued this line of approach and managed to obtain a canonical strong completion of a $T_0$ space.  Consequently, we showed that the category of strongly complete spaces is a reflective full subcategory of topological spaces and $SI^+$-continuous maps.  

Note that every sober space is a strong complete space.  But not every strong complete space is sober; indeed even a space which is a complete lattice with respect to its specialization order need not be sober (\cite{isbell82}).  At this moment, little is known about the relationship between order-theoretic properties (e.g., completeness conditions) of the specialization order and the sobriety of a space.  Because our approach makes essential use of irreducible subsets of a space and involves completeness conditions regarding the space, it might be a good starting point to carry out research with regards to the aforementioned question.

\end{document}